\documentclass[letterpaper, 10 pt, conference]{ieeeconf}  

\IEEEoverridecommandlockouts                              

\usepackage{acronym}
\usepackage{amsmath}
\usepackage{amsthm}
\usepackage{amssymb}
\usepackage{mathtools}
\usepackage{amsfonts}
\usepackage{algorithm}
\usepackage{algorithmicx}
\usepackage[noend]{algpseudocode}
\usepackage{csquotes}
\usepackage{paralist}
\usepackage{tikz} 
\usepackage{url}
\usetikzlibrary{arrows,automata}
\usetikzlibrary{shapes}
\usetikzlibrary{positioning}

\title{\LARGE \bf
Opportunistic Qualitative Planning in Stochastic Systems with  Preferences over Temporal Logic Objectives
}

\author{Abhishek Ninad Kulkarni$^{\ast}$ and Jie Fu
\thanks{A. N. Kulkarni and J. Fu are with the Dept. of Electrical and Computer Engineering, University of Florida, Gainesville, FL 32603 USA.
{\tt\small \{a.kulkarni2,fujie\}@ufl.edu}}
}

\newif\ifuseboldmathops
\newif\ifuseittextabbrevs
\useboldmathopstrue   

\ifuseittextabbrevs

	\newcommand{\ie}{{\it i.e.}}

\else

	\newcommand{\ie}{i.e.}

\fi

\ifuseboldmathops

\else

\fi

\ifuseboldmathops

\else

\fi

\ifuseboldmathops

\else

\fi

\ifuseboldmathops

\else

\fi

\newcommand{\Eventually}{\Diamond \, }
\newcommand{\Next}{\bigcirc \, }
\newcommand{\until}{\mbox{$\, {\sf U}\,$}}

\newcommand{\calF}{\mathcal{F}}
\newcommand{\calAP}{\mathcal{AP}}
\newcommand{\calPA}{\mathcal{B}}

\newcommand{\init}{{q_0}}

\acrodef{mdp}[MDP]{Markov Decision Process}
\acrodef{pomdp}[POMDP]{Partially Observable Markov Decision Process}
\acrodef{ltl}[LTL]{Linear Temporal Logic}
\acrodef{dfa}[DFA]{Deterministic Finite Automaton}

\theoremstyle{definition}
 \newtheorem{definition}{Definition}
 \newtheorem{example}{Example}

\newtheorem{lemma}{Lemma}

\newtheorem{remark}{Remark}
\newtheorem{theorem}{Theorem}

\newcommand{\calA}{\mathcal{A}}

\newcommand{\Mod}{\mathsf{Mod}}
\newcommand{\augnodes}{\mathcal{W}}
\newcommand{\augnode}{W}

\newcommand{\augedges}{\mathcal{E}}
\newcommand{\asw}{\mathsf{ASWin}}

\acrodef{smdp}[Semi-MDP]{Semi-Markov decision process}
\acrodef{mcts}[MCTS]{Monte Carlo tree search}
\acrodef{uct}[UCT]{Upper Confidence Bound 1 applied to trees}
\acrodef{scltl}[scLTL]{syntactically co-safe LTL}
\acrodef{ssp}[SSP]{Stochastic Shortest Path}
\acrodef{p2sg}[SG(2)]{Two-player Stochastic Game}
\acrodef{mc}[MC]{Markov chain}
\acrodef{prefltl}[TPL]{ Temporal Preference Logic}
\acrodef{tld}[TLwD]{Temporal Logic with Distributions}
\acrodef{mtl}[Metric TL]{Metric Temporal Logic}
\acrodef{sta}[STA]{Stochastic Timed Automaton}

\newcommand{\dist}{\mathcal{D}}

\renewcommand{\Pr}{\mathbf{Pr}}

\newcommand{\calM}{\mathcal{M}}

\acrodef{gpf}[GPF]{generalized preference formula}

\acrodef{cp}[CP]{ceteris paribus}
\acrodef{milp}[MILP]{Mixed-Integer Linear Programming}
\acrodef{dfa}[DFA]{Deterministic Finite Automaton}

\newcommand{\weakpref}{\trianglerighteq}

\newcommand{\prefnodes}{\mathcal{X}}
\newcommand{\prefedges}{E}

\newcommand{\outcomes}{\mathsf{Outcomes}}

\newcommand{\prefix}{\mathsf{Pref}}

\begin{document}

\maketitle
\thispagestyle{empty}
\pagestyle{empty}

\begin{abstract}
Preferences play a key role in determining what goals/constraints to satisfy when not all constraints can be satisfied simultaneously. In this work, we study preference-based planning in a stochastic system modeled as a Markov decision process, subject to a possible incomplete preference over temporally extended goals. Our contributions are three folds: First, we introduce a preference language  to specify  preferences over temporally extended goals. Second, we define a novel automata-theoretic model to represent the preorder induced by given preference relation.
The automata representation of preferences enables us to  develop a preference-based planning algorithm for stochastic systems.
Finally, we   show how to synthesize opportunistic strategies that achieves an outcome that improves upon the current satisfiable outcome, with positive probability or with probability one, in a stochastic system. We illustrate our solution approaches using a robot motion planning example. 

\end{abstract}

\section{Introduction}

Preference-based planning decides what constraints
to satisfy when not all constraints can be achieved
\cite{hastie2010rational}. 
In this paper, we study a class of 
qualitative, preference-based probabilistic planning problem in which the agent aims to strategically exploit the \emph{opportunities} that arise due to stochasticity in its environment to achieve a more preferred outcome than what may be achieved from its initial state. 
Such problems are encountered in many applications of autonomous systems.

In existing methods for probabilistic planning with temporal goals, the desired behavior of the system is specified by a temporal logic formula \cite{manna2012temporal}, and the goal is to compute a policy that either maximizes the probability of satisfying the formula \cite{ding2011mdp,hasanbeig2019reinforcement}, or enforces the satisfaction as a constraint  \cite{Lacerda2014,wen2015correct}. In recent work, preference-based planning with temporal logic objectives have been studied: minimum violation planning in a deterministic system \cite{tumova2013least} decides which low-priority constraints to be violated.  
Automated specification-revision is proposed in \cite{Lahijanian2016} where the formula can be revised with a cost and the planning problem is formulated into a multi-objective \ac{mdp} that trades off minimizing the cost of revision and maximizing the probability of satisfying the revised formula.  \cite{mehdipourSpecifyingUserPreferences2021} introduced weights  with Boolean and temporal operators in signal temporal logic to specify the importance of satisfying the subformula and priority in the timing of satisfaction.  They developed a gradient-based optimization method to maximize the weighted satisfaction in deterministic dynamical systems. 
Robust and recovery specifications are introduced by \cite{bloem2019synthesizing} and pre-specify what behaviors are expected when the part of the system specification (\ie,  the environment assumption) fails to be satisfied.  Existing  preference-based planning methods with temporal goals assume the preference relation to be \emph{complete}.  

Unfortunately, in many  applications, the completeness assumption does not always hold. For instance, it can be impractical to elicit user's preference between every pair of outcomes when the set of outcomes is large; or in some situation, such as the trolley problem \cite{thomson1976killing}, the outcomes (sacrificing passengers or pedestrians) are   incomparable. 
Preference languages have been proposed to represent both the complete and incomplete preferences over propositional formulas  \cite{van2005preference} and temporal logic formulas \cite{bienvenuPreferenceLogicsPreference2010}. For planning, CP-net and its variants \cite{santhanam2016representing,boutilier2004cp} have been proposed as a computational model. But they are defined over propositional preferences. To the best of our knowledge, there is no computational model that can express incomplete preferences over temporal goals. Such a model is needed to facilitate planning in stochastic environments. 


In this paper, we propose a novel automata-theoretic approach to qualitative planning in \ac{mdp}s with incomplete preferences over temporal logic objectives. 
Our approach consists of three steps. First, we express (incomplete) preferences over the satisfaction of  temporal goals specified using a fragment of \ac{ltl}.
Unlike propositional preferences that are interpreted over states, preferences over temporal goals are interpreted over infinite words. Second, we define an \emph{automata-theoretic model} to represent the preorder induced by the preference relation and describe a procedure to construct the automata-theoretic model given a preference formula. 
Thirdly, we present an algorithm to solve  preference-based strategies in a stochastic system modeled as a labeled \ac{mdp}. We presented   \emph{safe and positively improving} and \emph{safe and almost-surely improving} strategies, that identify and exploit opportunities for improvements with positive probability and probability one, respectively. A running example is employed to illustrate the notions and solution approaches.


\section{Preliminaries}
	\label{sec:prelim}


\textbf{Notation.} Given a finite set $X$, let $\dist(X)$ be the set of probability distributions over $X$. 
Let $\Sigma$ be an alphabet (a finite set of symbols). We denote the set of finite (resp., infinite) words that can be generated using $\Sigma$ by $\Sigma^\ast$ (resp., $\Sigma^\omega$). 
Given a word $w\in \Sigma^\omega$, a prefix of $w$ is a word $u \in \Sigma^\ast$ such that there exists $v \in \Sigma^\omega$, $w=uv$.
We denote the set of all finite prefixes of $w$ by $\prefix(w)$. 

We consider a class of decision-making problems in  stochastic systems modeled as a labeled \ac{mdp} \cite{baier2008principles}. 

\begin{definition}[Labeled \ac{mdp}]
	\label{def:labeled_mdp}
	A labeled \ac{mdp} is a tuple $M = \langle S, A, P, \calAP, L \rangle,$ where $S$ and $A$ are finite state and action sets, $P: S \times A \rightarrow \dist(S)$ is the transition probability function such that $P(s' \mid s, a)$ is the probability of reaching $s' \in S$ given that action $a \in A$ is chosen at state $s \in S$, $\calAP$ is a finite set of atomic propositions, and $L: S \rightarrow 2^{\calAP}$ is a labeling function that maps each state to a set of atomic propositions which are true in that state.
\end{definition}

A finite-memory, randomized strategy in the \ac{mdp} is a function $\pi : S^\ast \rightarrow \dist(A)$. A Markovian, randomized strategy in the \ac{mdp} is a function $\pi : S \rightarrow \dist(A)$. Given an \ac{mdp} $M$ and an initial distribution $\nu_0$, a   strategy $\pi$ induces a stochastic process $M_\pi = \{S_t \mid t \geq 1\}$ where $S_k$ is the random variable for the $k$-th state in the stochastic process $M_\pi$ and it holds that $S_0 \sim \nu_0$ and $S_{i+1} \sim P(\cdot \mid S_i, a_i)$ and $a_i \sim \pi(\cdot \mid S_0\ldots S_i)$ for   $i \geq 0$.

We express the objective of the planning agent as preferences over a set of outcomes, each of which is   expressed by a \ac{scltl} formula \cite{kupferman2001model}.

\begin{definition}
	Given a set of atomic propositions $\calAP$, an \ac{scltl} formula is defined inductively as follows:
	\[
	\varphi \coloneqq p \mid \neg p \mid \varphi \land \varphi \mid  \Next \varphi \mid \varphi \until \varphi,
	\]
	where $p \in \calAP$ is an atomic proposition. The operators $\neg$ (negation) and $\land$ (and) are propositional logic operators. The operators $\Next$ (next) and $\until$ (until) are temporal operators \cite{kupferman2001model}. The operator $\Eventually$ (eventually) is derived using $\until$ as follows: $\Eventually \varphi = \top \until \varphi$ where $\top$ is unconditionally true. The formula $\Eventually \varphi$ is true if $\varphi$ holds in some future time.
	
\end{definition}

The \ac{scltl} formulas are a subclass of \ac{ltl} formulas with a special property that an infinite word satisfying an \ac{scltl} only needs to have a `good' prefix (formalized after Definition~\ref{def:dfa}). The set of good prefixes can be compactly represented as the language accepted by a \ac{dfa}.  

\begin{definition}
	\label{def:dfa}
	A deterministic finite automaton (DFA) is a tuple 
$\calA = \langle Q, \Sigma, \delta, \init, F \rangle,$
	where $Q$ is a finite set of states; $\Sigma = 2^\calAP$ is a finite set of symbols called the alphabet; $\delta: Q \times \Sigma \rightarrow Q$ is a deterministic transition function that maps a state and a symbol to a next state. The transition function is extended recursively over words as follows: $\delta(q,\sigma u)=\delta(\delta(q, \sigma),u )$ given $\sigma \in \Sigma$ and $u \in \Sigma^\ast$; $\init \in Q$ is the initial state; $F \subseteq Q$ is a set of accepting states. A word $u$ is accepted by $\calA$ if $\delta(\init, u)\in F$. 
\end{definition}

Given an \ac{scltl} formula $\varphi$ and an infinite word $w \in \Sigma^\omega$, a `good' prefix is a finite word $u \in \Sigma^\ast$ such that $u \in \prefix(w)$ and $u$ is accepted by the \ac{dfa}, $\calA$. A   word $w \in \Sigma^\omega$ satisfies an \ac{scltl} formula $\varphi$, denoted by $w \models \varphi$, if $w$ has a good prefix.  The set of words satisfying an \ac{scltl} formula $\varphi$ is denoted by $\Mod(\varphi) = \{w \in \Sigma^\omega \mid w \models \varphi \}$. For an \ac{scltl} formula,  all accepting states of its corresponding \ac{dfa} are absorbing, \ie, $\delta(q, \sigma) = q$ for any $q \in F$ and $\sigma \in \Sigma$. We assume the transition function of \ac{dfa} to be \emph{complete}. That is, $\delta(q,\sigma)$ is defined for any pair $(q, \sigma) \in Q \times \Sigma$. An incomplete transition function can be made complete by introducing a sink state and redirecting all undefined transitions to that sink state.

An infinite path in a labeled \ac{mdp} $\rho=s_0 s_1 \ldots$ induces a word $w = L(s_0)L(s_1)\ldots$ in the \ac{dfa}.
We say the path $\rho$ satisfies an \ac{scltl} formula $\varphi$ if and only if the induced word $w$ satisfies the formula, \ie, $w \models \varphi$. 

\begin{definition}[Almost-Sure/Positive Winning Strategy]
	Given an \ac{mdp} $M$ and an \ac{scltl} formula $\varphi$, a  strategy $\pi: S^\ast \rightarrow \dist(A)$ is said to be almost-sure (resp., positive) winning if, in the stochastic process $M_\pi$  induced by $\pi$, the formula $\varphi$ can be satisfied 
	with probability one (resp. with a probability $\ge 0$). Formally, in the stochastic process $M_\pi=\{S_t\mid t\ge 1\}$, $\Pr(S_0S_1,\ldots \models \varphi) =1$ (resp. $>0$).
\end{definition}

The set of \emph{states} in the \ac{mdp} $M$, starting from which the agent has an almost-sure (resp. positive) winning strategy to satisfy an \ac{scltl} formula $\varphi$ is called the \emph{almost-sure (resp., positive) winning region}. Given an \ac{mdp} and an \ac{scltl} formula, the product operation \cite{baierControllerSynthesisProbabilistic2004} reduces the problem of computing almost-sure (resp., positive) winning region to that of computing the set of states from which a subset of final states in the product \ac{mdp} can be reached with almost-surely (resp., positive probability). It is known that there exists a \emph{memoryless, almost-sure winning strategy} $\pi$ to ensure the subset of final states is reached with probability one from a state in the almost-sure winning region. Polynomial (resp., linear) time algorithm to compute almost-sure (resp., positive) winning strategy in \ac{mdp}s with reachability objectives can be found in the book by \cite[Chap. 10]{baier2008principles}.

\subsection{Running Example}
\label{sec:running-example}

We use a motion planning problem for an cleaning robot to illustrate the the concepts discussed in this paper. The robot is to operate in a $5 \times 5$ stochastic gridworld as shown in Figure~\ref{fig:gridworld}. At every step, the robot must choose to move in one of the \texttt{North, East, South, West} directions. If the action results in an obstacle cell (shown in black), the robot returns to the cell it started from. If the robot enters a cell marked with green arrows, it may either stay in that cell or move into an adjacent cell along a direction indicated by the arrows each with a positive probability. If the robot moves into any cell with no arrows, it remains in that cell with probability one. The robot has a limited battery capacity measured in units. Every action costs $1$ unit of battery. We consider two preferences objectives for the robot. 

\begin{figure}
	\centering
	\includegraphics[scale=0.275]{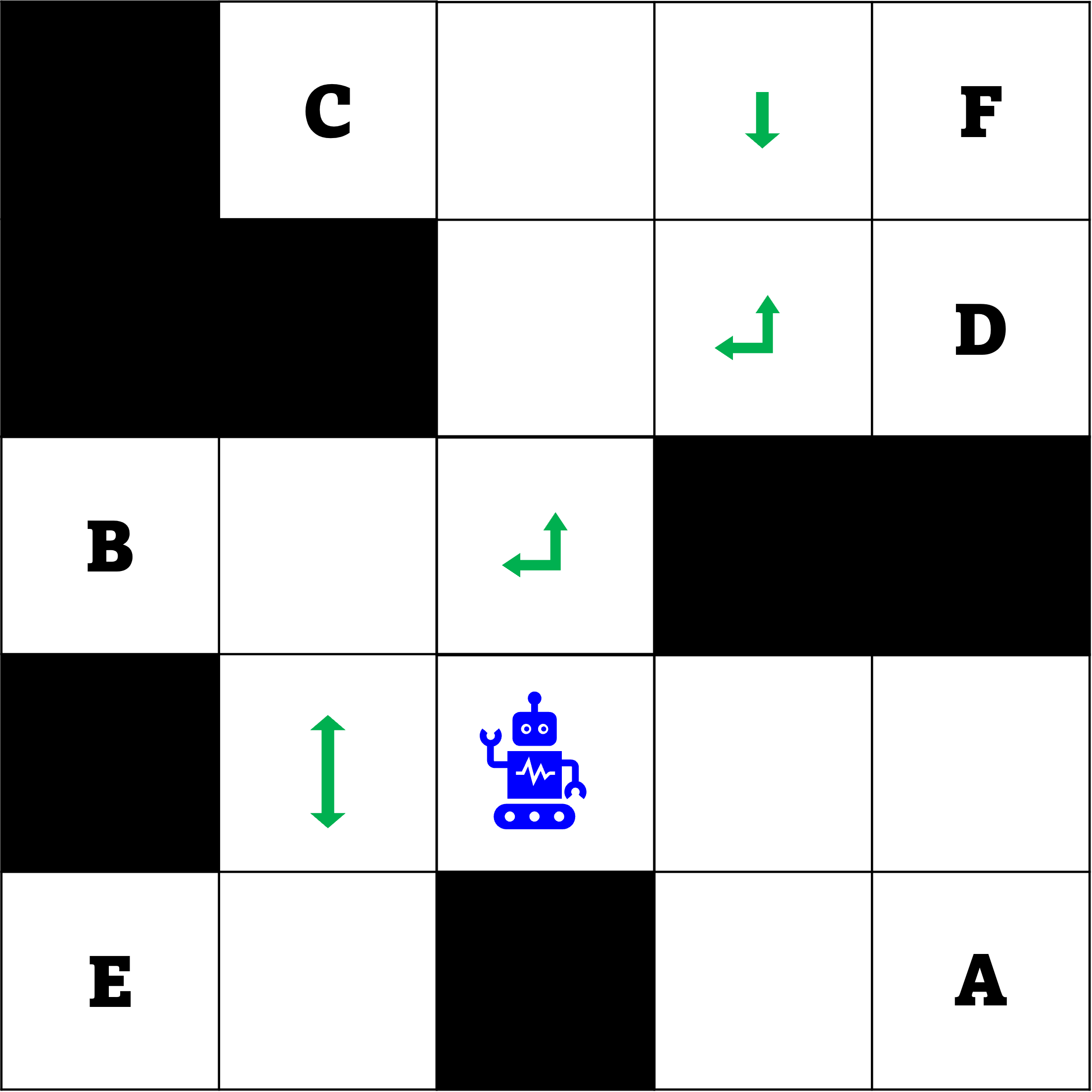}	

	\caption{A gridworld \ac{mdp} with $6$ regions of interest $A$-$F$.}
	\label{fig:gridworld}
\end{figure}

\begin{enumerate}
    \item[(PO1)] The robot must visit $A, B$ and/or $E$, given the preference that: visiting $B$ is strictly preferred to visiting $A$, and visiting $E$ is strictly preferred to visiting $A$. 
    
    \item[(PO2)] The robot must visit exactly one of $A, B, C, D$ or $F$, given the preference that: visiting $B$ is strictly preferred to visiting $A$, 
    visiting $D$ is strictly preferred to visiting $B$, visiting $F$ is strictly preferred to visiting $C$, and visiting $B$ is indifferent to visiting $C$. 
\end{enumerate}

The preference relations expressed by both the objectives are incomplete. In the first objective, neither the relation between $B$ and $E$ is given nor can it be deduced using the properties (e.g., transitivity) of preferences. Hence, visiting $B$ and visiting $E$ are incomparable outcomes due to incompletely known preferences.

In the second objective, since $B$ and $C$ are indifferent, it follows by transitivity that visiting $C$ is strictly preferred to visiting $A$, and visiting $D$ is strictly preferred to visiting $C$. However, visiting $D$ is incomparable to visiting $F$ since no relation is either given or can be deduced between them.

\section{Preference Modeling} 
\label{sec:language}

In this section, we propose a language to compactly represent incomplete preferences over temporal goals. 


Let $\Phi = \{\varphi_1, \ldots, \varphi_n\}$ be an indexed set of outcomes, \ie, temporal goals expressed by \ac{scltl} formulas.

\begin{definition}
    A preference   on $\Phi$ is a reflexive binary relation $\weakpref$ on $\Phi$. For any $1 \leq i, j \leq n$, a pair of outcomes $(\varphi_i, \varphi_j) \in~\weakpref$ means that satisfying $\varphi_i$ is considered ``at least as good as'' satisfying $\varphi_j$. 
\end{definition}

We also denote $(\varphi_i, \varphi_j) \in~\weakpref$ by $\varphi_i \weakpref \varphi_j$. Given any pair of outcomes, $\varphi_i, \varphi_j \in \Phi$, exactly one of the following four relations holds: 
\begin{enumerate}
    \item $\varphi_i$ is \emph{indifferent} to $\varphi_j$: $\varphi_i \weakpref \varphi_j$ and $\varphi_j \weakpref \varphi_i$,
        
    \item $\varphi_i$ is \emph{strictly preferred} to $\varphi_j$: $\varphi_i \weakpref \varphi_j$ and $\varphi_j \not\weakpref \varphi_i$,

    \item $\varphi_j$ is \emph{strictly preferred} to $\varphi_i$: $\varphi_j \weakpref \varphi_i$ and $\varphi_i \not\weakpref \varphi_j$,
    
    \item $\varphi_i$ is \emph{incomparable} to $\varphi_j$: $\varphi_i \not\weakpref \varphi_2$ and $\varphi_j \not\weakpref \varphi_i$.
\end{enumerate}


When the agent is indifferent to two outcomes $\varphi_i, \varphi_j$, it may choose to satisfy either one of them. This can equivalently be expressed in \ac{scltl} by the disjunction of the two formulas. Based on this observation, we hereby assume that for any two outcomes $\varphi_i, \varphi_j \in \Phi$, $\varphi_i \weakpref \varphi_j$ and $\varphi_j \weakpref \varphi_i$ do not hold simultaneously, \ie, no two outcomes in $\Phi$ are indifferent to each other.
%
%
As a result, the binary relation $\weakpref$ on $\Phi$ can equivalently be expressed using the two sets $P, J \subseteq \Phi \times \Phi$ constructed as follows: given a pair of outcomes $\varphi_i, \varphi_j \in \Phi$, $1 \leq i, j, \leq n$,
\begin{enumerate}
    \item $(\varphi_i, \varphi_j) \in P$ iff $\varphi_i$ is strictly preferred to $\varphi_j$, 
    
    \item $(\varphi_i, \varphi_j) \in J$ iff $\varphi_i$ is incomparable to $\varphi_j$.
\end{enumerate}

\begin{remark}
    We closely follow the notation in \cite[Ch. 2]{bouyssou2009concepts}. In contrast, we use the properties of \ac{scltl} formulas to simplify the notation to avoid expressing indifference explicitly.  
\end{remark}


Notice that the sets $P, J$ induce a mutually exclusive and exhaustive partition of $\Phi \times \Phi$. Let $P^- = \{(\varphi_j, \varphi_i) \in \Phi \times \Phi \mid (\varphi_i, \varphi_j) \in P\}$. Then, $P \cup P^- \cup J = \Phi \times \Phi$ and $P \cap P^- = P^- \cap J = J \cap P = \emptyset$. 

\begin{example}
    Consider the running example introduced in Sect.~\ref{sec:running-example}. In preference objective (PO1), since there is no constraint on visiting multiple regions of interests, each outcome can be represented using ``eventually'' operator. Hence, the set of outcomes is given by $\Phi = \{\Eventually A, \Eventually B, \Eventually E\}$. The components of preference structure $\langle P, J \rangle$ are given as follows: $P = \{(\Eventually B, \Eventually A), (\Eventually E, \Eventually A)\}$, and $J = \{(\Eventually B, \Eventually E), (\Eventually E, \Eventually B)\}$. 
    
    
    In preference objective (PO2), since exactly one region is to be visited, the outcomes can be represented as \ac{scltl} formulas: $\varphi_A = \neg (B \lor C \lor D \lor F) \until A$, $\varphi_B = \neg (A \lor C \lor D \lor F) \until B$, and so on. Because of the indifference, we replace $\varphi_B$ and $\varphi_C$ by their disjunction, $\varphi_B \lor \varphi_C$. Hence, the set of outcomes is $\Phi = \{\varphi_A, \varphi_B \lor \varphi_C, \varphi_D, \varphi_F\}$. And, the components of preference structure are given by: $P = \{(\varphi_B \lor \varphi_C, \varphi_A), (\varphi_D, \varphi_B \lor \varphi_C), (\varphi_F, \varphi_B \lor \varphi_C), (\varphi_D, \varphi_A), (\varphi_F, \varphi_A)\}$, and $J = \{(\varphi_F, \varphi_D), (\varphi_D, \varphi_F)$. 
    
\end{example}

Because an \ac{scltl} formula is interpreted over infinite words, the preference structure $\weakpref$ induces a preference structure $\succeq$ on the set of infinite words in $\Sigma^\omega$.  
%
%
Therefore, we can define a pre-order $\succeq \in \Sigma^\omega \times \Sigma^\omega$ based on the preference structure $\weakpref$ (and equivalently to the tuple $\langle P, J\rangle$). This is a non-trivial task  because any word in $\Sigma^\omega$ could satisfy more than one of the \ac{scltl} formulas in $\Phi$. Thus, to determine whether a word is strictly preferred over another, we need a way to compare two arbitrary subsets of $\Phi$ that contain outcomes satisfied by these two words. 

\begin{definition}[Most-Preferred Satisfied Outcomes]
	Given a word $w \in \Sigma^\omega$, let $\outcomes(w) = \{\varphi \in \Phi \mid w \models \varphi\}$ be the set of outcomes satisfied by $w$. 
	Given a subset $\Psi\subseteq \Phi$, let $\mathsf{MP}(\Psi) =\{\varphi \in \Psi\mid \nexists \varphi'\in \Psi: (\varphi, \varphi') \in P\}$ and 
	let $\mathsf{MP}(w) = \mathsf{MP}(\outcomes(w))$ 
	be the set of \emph{most-preferred outcomes} satisfied by the word $w$. 
\end{definition}


\begin{lemma}
Given a word $w\in \Sigma^\omega$, 	any pair  $\varphi, \varphi'\in \mathsf{MP}(w)$ are incomparable to each other. 
\end{lemma}
The proof follows from the definition.

\begin{definition}[Semantics]
\label{def:semantics}
    Given two words $w_1, w_2 \in \Sigma^\omega$, $w_1$ is strictly preferred to $w_2$, denoted $w_1\succ w_2$,  if and only if the following conditions hold:
    \begin{inparaenum}
        \item there exist $\varphi_i \in \mathsf{MP}(w_1)$ and $\varphi_j \in \mathsf{MP}(w_2)$ such that $(\varphi_i, \varphi_j) \in P$, and
        \item for every pair $\varphi_i \in \mathsf{MP}(w_1)$ and $\varphi_j \in \mathsf{MP}(w_2)$, $(\varphi_j, \varphi_i) \notin P$.
    \end{inparaenum}
    Word  $w_1$ is indifferent to $w_2$, denoted $w_1\sim w_2$,  if and only if $\mathsf{MP}(w_1) = \mathsf{MP}(w_2)$. Two words $w_1$ and $w_2$ are incomparable, denoted $w_1 \| w_2$, if neither $w_1 \succ w_2$, nor $w_2 \succ w_1$, nor $w_1 \sim w_2$ holds. 
\end{definition}

In words, $w_1$ is strictly preferred to $w_2$ iff: first, $w_1$ satisfies at least one \ac{scltl} formula that is strictly preferred to some \ac{scltl} formula satisfied by $w_2$. Second, every \ac{scltl} formula satisfied by $w_1$ is either strictly preferred to, or incomparable to any \ac{scltl} formula satisfied by $w_2$.

\begin{example}
    Consider preference objective (PO2). 
    %
    Consider two paths $\rho_1, \rho_2$   in Fig.~\ref{fig:gridworld} that sequentially visit  $A, F, D$ and $A, C$, respectively. Let $w_1 = L(\rho_1)$, $w_2 = L(\rho_2)$ be the words induced by $\rho_1, \rho_2$, respectively. For the word $w_1$, we have $\outcomes(w_1) = \{\varphi_A, \varphi_D, \varphi_F\}$ and $\mathsf{MP}(w) = \{\varphi_D, \varphi_F\}$ since visiting $D$ and $F$ individually is strictly preferred to $A$, and visiting $D$ and visiting $F$ are incomparable. Similarly, $\mathsf{MP}(w_2) = \{\varphi_C \lor \varphi_B\}$. Therefore, we have $w_1 \succ w_2$ because, condition (1) of strict preference semantics holds for the pair $(\varphi_D, \varphi_C \lor \varphi_B)$ and, condition (2) is also satisfied because $\varphi_F$ is incomparable to $\varphi_C \lor \varphi_B$. 
\end{example}

\section{Automata-Theoretic Computational Model for Incomplete Preferences}  
\label{sec:computational-model}

We now introduce a novel automata-theoretic computational model called a preference \ac{dfa}.

\begin{definition}[Preference \ac{dfa}]
	\label{def:paut}
	A preference \ac{dfa} is the tuple 
	\[
	    \calPA = \langle Q, \Sigma, \delta, \init, F, G \rangle,
	\]
	where $Q, \Sigma, \delta, \init$ are the (finite) set of states, the alphabet, the deterministic transition function, and an initial state, similar to these components in a \ac{dfa}. $F \subseteq Q$ is a set of final states. The last component $G = (\prefnodes, \prefedges)$ is a preference graph, where each node $X \in \prefnodes$ represents a subset of  final states $F$ such that $X_i \cap X_j = \emptyset$ for every $X_i, X_j \in \prefnodes$. The edges $\prefedges \subseteq \prefnodes \times \prefnodes$ is a set of directed edges. 
\end{definition}


Intuitively, a preference \ac{dfa} $\calPA$ encodes the preference relation $\succeq$ over subsets of words (languages in $\Sigma^\omega$) represented using different classes by defining a preorder over the acceptance conditions (sets of final states). Next, we describe construction a preference \ac{dfa} from a preference structure.

Given a preference structure $\weakpref = \langle P, J \rangle$, the preference \ac{dfa} is constructed in two steps. First, the underlying \ac{dfa} $\langle Q, \Sigma, \delta, \init, F \rangle$ is constructed as a cross product of \ac{dfa}s representing the union of languages of all \ac{scltl} formulas in $\Phi$. Letting $\calA_i = \langle Q_i, \Sigma, \delta_i, \init_i, F_i\rangle$ to be the \ac{dfa} corresponding to $\varphi_i$ for all $1 \leq i \leq n$, we have $Q = Q_1 \times \ldots \times Q_n$, $\delta(q, \sigma) = (\delta_1(q_1, \sigma), \ldots, \delta_n(q_n, \sigma))$, $\init = (\init_1, \ldots, \init_n)$ and $F = (F_1 \times Q_2 \times \ldots \times Q_n) \cup (Q_1 \times F_2 \times \ldots \times Q_n) \cup \ldots \cup (Q_1 \times Q_2 \times \ldots \times F_n)$. By definition, any word that induces a visit to a final state in preference automaton achieves at least one outcome in $\Phi$.

In the second step, we construct the preference graph $G = (\prefnodes, \prefedges)$. Intuitively, every node of the preference graph represents an equivalence class of final states such that any two words that visit any final state represented by the same node are indifferent under $\succeq$. To define the nodes, we first associate each final state with a set of \emph{tags}:

\begin{enumerate}
    \item A tag $x_{ij}$ is associated with a final state $q = (q_1, \ldots, q_n) \in F$ to denote that a word reaching $q$ satisfies a more preferred outcome among $\varphi_i$ and  $\varphi_j$. Hence, $x_{ij}$ is assigned to $q$ iff the following conditions hold:
    \begin{inparaenum}
        \item[(a)] $q_i \in F_i$, 
        \item[(b)] $(\varphi_i, \varphi_j) \in P$, 
        \item[(c)] $\varphi_i \in \mathsf{MP}(\{\varphi_k \in \Phi \mid q_k \in F_k\})$.     
    \end{inparaenum}
    
    \item A tag $y_{ij}$ is associated to a final state $q = (q_1, \ldots, q_n) \in F$ to denote that a word reaching $q$ satisfies the less preferred outcome among $\varphi_i$ and $\varphi_j$. Hence, $y_{ij}$ is assigned to $q$ iff:
    \begin{inparaenum}
        \item[(a)] $q_i \in Q_i \setminus F_i$, 
        \item[(b)] $q_j \in F_j$, 
        \item[(c)] $(\varphi_i, \varphi_j) \in P$, 
        \item[(d)] $\varphi_j \in \mathsf{MP}(\{\varphi_k \in \Phi \mid q_k \in F_k\})$.
    \end{inparaenum}
\end{enumerate}

We denote the set of tags associated to a final state $q \in F$ by $\lambda(q)$. A node $X \in \prefnodes$ represents a set of final states that have the same set of tags. That is, for any $q, q'\in X$, $\lambda(q)=\lambda(q')$. We write $\lambda(X) := \lambda(q)$ to denote the set of tags associated with any final state represented by $X$. An edge $(X_2, X_1)$ in $E$ represents that any final state in $X_1$ is strictly preferred to any final state in $X_2$. Thus, $(X_2, X_1)$ is included in $E$ if and only if
\begin{inparaenum}
    \item[(1)] there exists $1 \leq i, j \leq n$ such that $x_{ij} \in \lambda(X_1)$ and $y_{ij} \in \lambda(X_2)$; 
    \item[(2)] for all $1 \leq i, j \leq n$ such that $x_{ij} \in \lambda(X_1)$ and $y_{ij} \in \lambda(X_2)$ does not hold, $y_{ij} \in \lambda(X_1)$ and $x_{ij} \in \lambda(X_2)$ also does not hold. 
\end{inparaenum}

An edge $(X_1, X_2) \in \prefedges$ is intuitively understood as follows. Condition (1) states that there must exist a pair of \ac{scltl} formulas $\varphi_i, \varphi_j \in \Phi$ such that $(\varphi_i, \varphi_j) \in P$, and any word that visits   $X_2$ must satisfy $\varphi_i$ and any word that visits $X_1$ must satisfy $\neg \varphi_i \land \varphi_j$. Condition (2) asserts that the opposite of condition (1) should never hold. That is, there must not exist a pair of \ac{scltl} formulas $\varphi_i, \varphi_j \in \Phi$ such that $(\varphi_i, \varphi_j) \in P$, and any word that visits   $X_1$ satisfies $\varphi_i$ and any word that visits $X_2$ satisfies $\neg \varphi_i \land \varphi_j$. 



\begin{example}
    We describe the construction of preference \ac{dfa} for first preference objective (PO1). 
    The underlying \ac{dfa} of the preference \ac{dfa} for (PO1) is constructed as the union of \ac{dfa}s corresponding to $\Eventually A, \Eventually B, \Eventually C$, and is shown in Fig.~\ref{fig:po1-pref-dfa}. Every state in preference \ac{dfa} is annotated as a tuple $(a_i, b_j, e_k)$ where $i, j, k = 0, 1$. The subscript $i, j, k = 0$ means that corresponding region has been visited. Therefore, all states except $(a_1, b_1, e_1)$ are final states since at least one of the formulas is satisfied in all states but $(a_1, b_1, e_1)$. 
    
        	

    \begin{figure*}[t!]
        \centering
        \begin{tikzpicture}[->,>=stealth',shorten >=1pt,auto,node distance=3cm, scale = 0.37,transform shape,align=center]
	\node[state,ellipse,minimum width=10cm,minimum height=1cm] (0) {\huge $(a_1, b_1, e_1), \{\}$};
	\node[state,ellipse,minimum width=10cm,minimum height=1cm,accepting] (1) [below of=0] {\huge $(a_0, b_1, e_1), \{y_{EA}, y_{BA}\}$};
	\node[state,ellipse,minimum width=10cm,minimum height=1cm,accepting] (2) [right=2cm of 0] {\huge $(a_1, b_1, e_0), \{x_{EA}\}$};
	\node[state,ellipse,minimum width=10cm,minimum height=1cm,accepting] (3) [below of=2] {\huge $(a_0, b_1, e_0), \{y_{BA}\}$};
	\node[state,ellipse,minimum width=10cm,minimum height=1cm,accepting] (4) [right=2cm of 2] {\huge $(a_1, b_0, e_1), \{x_{BA}\}$};
	\node[state,ellipse,minimum width=10cm,minimum height=1cm,accepting] (5) [below of=4] {\huge $(a_0, b_0, e_1), \{y_{EA}\}$};
	\node[state,ellipse,minimum width=10cm,minimum height=1cm,accepting] (6) [right=2cm of 4] {\huge $(a_1, b_0, e_0), \{x_{BA}, x_{EA}\}$};
	\node[state,ellipse,minimum width=10cm,minimum height=1cm,accepting] (7) [below of=6] {\huge $(a_0, b_0, e_0), \{x_{BA}, x_{EA}\}$};
	
	\path 
	(0) edge node {\huge $A$} (1)
	(2) edge node {\huge $A$} (3)
	(4) edge node {\huge $A$} (5)
	(6) edge node {\huge $A$} (7)
	
	(0) edge node {\huge $E$} (2)
	(4) edge node {\huge $E$} (6)
	(1) edge node {\huge $E$} (3)
	(5) edge node {\huge $E$} (7)
	
	(0) edge[out=20,in=160] node[above] {\huge $B$} (4)
	(2) edge[out=20,in=160] node[above] {\huge $B$} (6)
	
	(1) edge[out=340,in=200] node[above] {\huge $B$} (5)
	(3) edge[out=340,in=200] node[above] {\huge $B$} (7)
	;
\end{tikzpicture}
       \vspace{-2ex} \caption{Preference DFA representing preference objective (PO1).}
        \label{fig:po1-pref-dfa}
    \end{figure*}
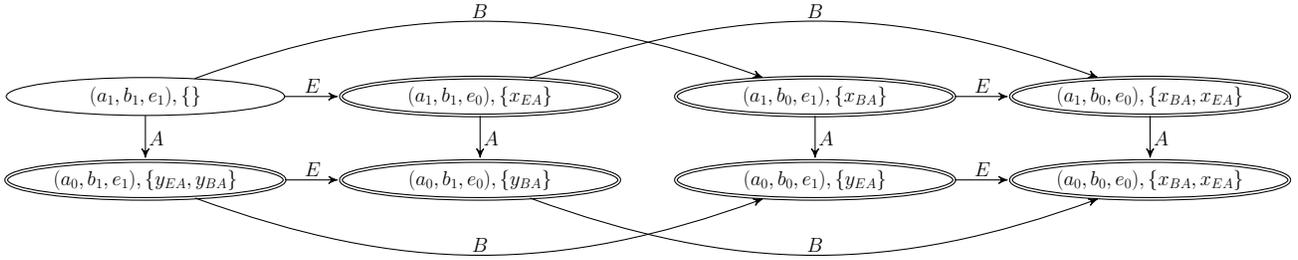
    
    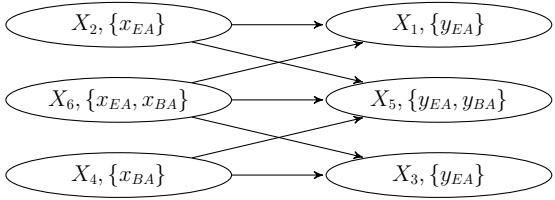
\begin{figure}
        \centering
	

\begin{tikzpicture}[->,>=stealth',shorten >=1pt,auto,node distance=2cm, scale = 0.5,transform shape,align=center]
	\node[state,ellipse,minimum width=6cm,minimum height=1cm] (0) {\LARGE $X_2, \{x_{EA}\}$};
	\node[state,ellipse,minimum width=6cm,minimum height=1cm] (1) [right=2.5cm of 0] {\LARGE $X_1, \{y_{EA}\}$};
	\node[state,ellipse,minimum width=6cm,minimum height=1cm] (2) [below of=0] {\LARGE $X_6, \{x_{EA}, x_{BA}\}$};
	\node[state,ellipse,minimum width=6cm,minimum height=1cm] (3) [right=2.5cm of 2] {\LARGE $X_5, \{y_{EA}, y_{BA}\}$};
	\node[state,ellipse,minimum width=6cm,minimum height=1cm] (4) [below of=2] {\LARGE $X_4, \{x_{BA}\}$};
	\node[state,ellipse,minimum width=6cm,minimum height=1cm] (5) [right=2.5cm of 4] {\LARGE $X_3, \{y_{EA}\}$};
	
	\path 
	(0) edge node {} (1)
	(2) edge node {} (1)
	(0) edge node {} (3)
	(2) edge node {} (3)
	(4) edge node {} (3)
	(2) edge node {} (5)
	(4) edge node {} (5)
	;
\end{tikzpicture}
               \vspace{-2ex}
 \caption{Preference graph corresponding to preference DFA in Fig.~\ref{fig:po1-pref-dfa}. }
        \label{fig:po1-pref-graph}
        \vspace{-3ex}
    \end{figure}
    
    Each final state is assigned a set of labels. For instance, the state $\lambda((a_1, b_0, e_0)) = \{x_{BA}, x_{EA}\}$ \footnote{We use $A,B,E$ in places of numerical indices. } since by any word that visits the state satisfies $\Eventually B$ and $\Eventually E$. This results in $6$ unique labels corresponding to a different class of equivalent words in $\Sigma^\omega$ that visit that final state. These classes form the nodes $X_k$ for $k = 1\ldots6$ of the preference graph shown in Fig.~\ref{fig:po1-pref-graph}. An edge $(X_2, X_5)$ expresses that any word that visits $X_5$ is strictly preferred to any word that visits $X_2$ but not $X_5$. Similarly, any word that visits  $X_4$ only  is incomparable to any word that visits $X_6$ only.
\end{example}




\section{Opportunistic Qualitative Planning with Incomplete Preferences}  
\label{sec:planning}

In this section, we define two types of strategies,
that exploit the \emph{opportunities} that arise due to stochasticity with a positive probability or with probability one, respectively.

\begin{definition}[Product of an \ac{mdp} with a Preference \ac{dfa}] 
	\label{def:product-mdp} 
	Given an \ac{mdp} $M=\langle S, A, P, \calAP, L\rangle$ and the preference \ac{dfa} $\calPA = \langle Q, \Sigma, \delta, \init, F, G = (\prefnodes, \prefedges) \rangle$, the product of \ac{mdp} with preference \ac{dfa} is defined as the tuple,
	\[
	\calM := \langle  V, A, \Delta, \calF, \mathcal{G} \rangle,
	\] 
	where $V \coloneqq S \times Q$ is the finite set of states. $A$ is the same set of actions as $M$. The transition function $\Delta: V \times A \rightarrow \dist(V)$ is defined as follows: for any states $(s, q), (s', q') \in V$ and any action $a \in A$, $\Delta((s',q') \mid (s, q), a) = P(s' \mid s, a)$ if $q' \in \delta(q, L(s'))$ and $0$ otherwise. $\calF \subseteq V$ is the set of final states by reaching which at least some outcome is achieved. The component $\mathcal{G} = (\augnodes, \augedges)$ is a graph where $\augnodes := \{S \times X \mid X \in \prefnodes\}$ is the set of nodes and $\augedges$ is a set of edges such that, for any $\augnode_i = S \times X_i$ and $\augnode_j = S \times X_j$,  $(\augnode_i, \augnode_j) \in \augedges$ if and only if $(X_i,X_j)\in \prefedges$. 
\end{definition}

In the product construction,
an edge $(W_i, W_j) \in \augnodes$ denotes that any path $\rho \in V^\ast$ that reaches  $W_j$ is strictly preferred to any path $\rho' \in V^\ast$ that reaches $W_i$ but not $W_j$ under the given preference. Thus, we transform the preference over words given by the preference \ac{dfa} to a preference over outcomes, each of which reaches a subsets of states in $\augnodes$.
%
%
For each node   $\augnode \in \augnodes$, we can compute a set of states, denoted $\asw(\augnode)$, from which the agent has a strategy to reach $\augnode$ with probability one, using the solution of almost-sure winning in \ac{mdp}s with reachability objective \cite{baier2008principles}.  It is possible that $v\in \asw(\augnode)$ and $v\in \asw(\augnode')$ where $\augnode\ne \augnode'$ and $(\augnode,\augnode')\in \augedges$. In this case, a preference satisfying strategy must visit the   preferred node $\augnode'$. To generalize, 
let $Z\subseteq \augnodes$ be a subset of nodes in the product, we overload the notation $\mathsf{MP}$ such that $\mathsf{MP}(Z) =\{\augnode \in Z\mid \nexists \augnode'\in Z, (\augnode,\augnodes')\in \augedges\}$. A preference satisfying strategy from $v$ must visit a node in $\mathsf{MP}(Z_v)$ where $Z_v= \{\augnode\in \augnodes\mid v\in \asw(\augnode)\}$.

However, at some states, the uncertainty may create opportunities to transition from the state $v$ to $v' \in V$ such that a more preferred node can be reached almost-surely from $v'$.
We call such a transition to be an \emph{improvement}.   



\begin{definition}[Improvement]
\label{def:improvement}
    Given any states $v_1, v_2 \in V$, $v_2$ is said to be an \emph{improvement} over $v_1$ if and only if there exists a pair of preference nodes $W_1 \in \mathsf{MP}(\{W \in \augnodes \mid v_1 \in \asw(W)\})$ and $W_2 \in \mathsf{MP}(\{W \in \augnodes \mid v_2 \in \asw(W)\})$ such that $(W_1, W_2) \in \prefedges$.
\end{definition}

 A transition from state $v \in V$ to $v' \in V$ is said to be \emph{improving} if $v'$ is an improvement over $v$.




\begin{definition}
\label{def:spi-strategy}
    A strategy $\pi: V \rightarrow 2^A \cup \{\uparrow\}$ \footnote{$\pi(v)=\uparrow$ means the function $\pi$ is undefined at $v$.} is said to be \emph{safe and positively improving (resp., safe and almost-surely improving)} if, the following conditions hold for any state $v \in V$ such that $\pi(v) \neq \uparrow$:
    (a) there exists (resp., for all) a path $\rho$ in $\calM_{\pi}$ with $\rho[0] = v$ such that, for some $i \geq 0$, $\rho[i+1]$ is an improvement over $\rho[i]$;
    (b) there does not exist a path $\rho$ in $\calM_{\pi}$ with $\rho[0] = v$ such that, for some $i \geq 0$, $\rho[i]$ is an improvement over $\rho[i+1]$. 
\end{definition}

Intuitively, the SPI and SASI strategies exploit opportunities by inducing an improving transition with a positive probability and with probability one, respectively.

We now define a new model called \emph{an improvement \ac{mdp}} that differentiates the states reached by  improving transitions.

\begin{definition}[Improvement \ac{mdp}]
\label{def:improvement-mdp}
	Given a product \ac{mdp} $\calM$, an \emph{improvement \ac{mdp}} is the tuple,
	\[
	    \tilde M = \langle \tilde V, A, \tilde \Delta,  \tilde \calF \rangle,
    \] 
    where $\tilde V = \{(v, \top), (v, \bot) \mid v \in V\}$ is the set of states, $ \tilde \calF = \{(v, \top) \mid v \in V\}$ is the set of final states. An action $a \in A$ is enabled at a state $v \in \tilde V$ if and only if for for any $v'$ such that $\Delta(v,a,v')>0$, $v$ is \emph{not} an improvement over $v'$. 
    The transition function $\tilde \Delta: \tilde V \times A \rightarrow \dist(\tilde V)$ is defined as follows:
    For any $v\in V$, for an action $a$ \emph{enabled from} $v$, if $\Delta(v,a,v')>0 $ and $v'$ is an improvement from $v$, then  let $\tilde \Delta((v, \bot), a, (v', \top)) = \Delta(v, a, v')$. Else, if $\Delta(v,a,v')>0 $ and $v'$ is not an improvement from $v$, then let $\tilde \Delta((v, \bot), a, (v', \bot)) = \Delta(v, a, v')$ and $\tilde \Delta((v, \top),  a, (v', \bot)) = \Delta(v, a, v')$.
\end{definition}

\begin{theorem}
    The following statements hold for any state $v \in V$ in product \ac{mdp}. 
    \begin{enumerate}
        \item An SPI strategy at $v$ is a positive winning strategy in improvement \ac{mdp} at the state $(v, \bot)$ to visit $\tilde \calF$. 
        
        \item An SASI strategy at $v$ is an almost-sure winning strategy in improvement \ac{mdp} at the state $(v, \bot)$ to visit $\tilde \calF$. 
    \end{enumerate}
\end{theorem}
\begin{proof}[Proof (Sketch)]
    Statement (1). By construction, any action which induces a transition that violates condition (b) in Def.~\ref{def:spi-strategy} with positive probability is disabled in the improvement \ac{mdp}. Also, by construction, any final state in $\calF$ can only be reached by making an improvement. Hence, a positive winning strategy in improvement \ac{mdp} which visits $\tilde \calF$ satisfies condition (a) in Def.~\ref{def:spi-strategy}. The proof of statement (2) is similar to that of statement (1). 
\end{proof}

The SPI and SASI strategies may exploit multiple opportunities by inducing sequential improvements: Whenever the agent reaches a state $ (v,\top)\in \tilde V$, he will check if a SPI (or SASI) strategy exists for $(v,\bot)$. If yes, then the agent will carry out the SPI (or SASI) strategy. Otherwise, the agent will 
carry out the almost-sure winning strategy for one of the  most preferred and satisfied objective at $v$.

\begin{example}
    Consider the case when robot is at $(2, 1)$ with $4$ units of battery and is to satisfy (PO1). Although the robot cannot almost-surely visit either $B$ or $E$ individually, it can almost-surely visit one of $B, E$ by moving \texttt{West}. Since visiting both $B$ and $E$ is strictly preferred to visiting $A$, moving \texttt{West} is a safe and almost-surely improving strategy at $(2, 1)$. Instead of $4$ units, if the robot starts with $2$ units of battery, it can reach neither of $A, B$ or $C$ almost-surely. In this case, the SASI strategy is undefined. The SPI strategy is to choose \texttt{West} because, with positive probability, it leads to cells $(1, 2), (1, 0)$ with $1$ units of battery remaining. From these states, one of $B, E$ can be reached almost-surely.

    Consider the robot whose objective is (PO2) starting at the cell $(2, 1)$ with $4$ units of battery. From this state, only $A$ can be visited almost-surely. The SASI strategy at $v_0$ is to move \texttt{North} because, with positive probability, the robot would reach one of the cells---$(1, 2), (2, 2), (2, 3)$---with $3$ units of battery remaining. Since from each of these states at least one of $B, C, D, F$ can almost-surely be achieved, the robot almost-surely makes an improvement. Suppose the robot reaches $(2, 3)$ with $3$ units of battery. From this state, only visiting $C$ is almost-surely winning. However, the SASI strategy is to move \texttt{North} and then \texttt{East}, thereby ensuring a visit to either $F$ or $D$ with probability one. Hence, we see that SASI strategy not only plans for a single improvement, but it may also induces multiple sequential improvements.  
\end{example}


\section{Conclusion}
	\label{sec:conclude}
	In this work, we propose a language to specify incomplete preferences as a pre-order over temporal objectives. This allows us to synthesize qualitatively plans even when some outcomes are incomparable. We define two types of opportunistic strategies that strategically, and whenever possible, improve the   outcome they can achieve  sequentially. Our work provides a method for stochastic planning with incomplete preferences over  a subclass of temporal logic objectives. Building on this work, we consider a number of future directions: 1) we will consider a preference over temporal objectives that encompass more general \ac{ltl} properties such as safety, recurrence, and liveness; 2) we will build on the qualitative reasoning to study quantitative planning with such preference specifications. The later requires us to  jointly consider how well (the probability) an objective is satisfied and how preferred is the objective.

\bibliographystyle{IEEEtran}
\bibliography{refs}

\end{document}